\documentclass[aps,pra,twocolumn, showkeys]{revtex4-1}
\usepackage{amsmath,paralist,amsthm,comment,amssymb}
\usepackage{graphicx}

\usepackage{tikz}
\usetikzlibrary{matrix,arrows}
\usetikzlibrary{shapes.geometric}
\usetikzlibrary{patterns} % LATEX and plain TEX when using TikZ

\usepackage{pgfplots}

\usepackage{algorithm} 
\usepackage{algorithmic}

\newcommand{\nix}[1]{}

\newtheorem{theorem}{Theorem}

\newtheorem{lemma}[theorem]{Lemma}

\newcommand{\ba}{\begin{array}}
\newcommand{\ea}{\end{array}}
\newcommand{\ben}{\begin{eqnarray}}
\newcommand{\een}{\end{eqnarray}}

\newcommand{\be}{\begin{eqnarray*}}
\newcommand{\ee}{\end{eqnarray*}}

\begin{document}

\title{Efficient Decoding of Topological Color Codes}

\author{Pradeep Sarvepalli}
%\email[]{pradeep@phas.ubc.ca}
\author{Robert Raussendorf}
%\email[]{raussen@phas.ubc.ca}
\affiliation{Department of Physics and Astronomy, University of British Columbia, Vancouver V6T 1Z1, Canada }

% Copyright by Pradeep Sarvepalli and Robert Raussendorf
% All rights reserved. 

\date{November 3, 2011}

\begin{abstract}
Color codes are a class of topological quantum codes with a high error threshold and large set of transversal encoded gates, and are thus suitable for fault tolerant quantum computation in two-dimensional architectures. Recently, computationally efficient decoders for the color codes were proposed. We describe an alternate efficient iterative decoder for topological color codes, and apply it to the color code on hexagonal lattice embedded on a torus. In numerical simulations, we find an error threshold of 7.8\% for independent dephasing and spin flip errors.   

\end{abstract}


\maketitle
\section{Introduction}

Topological quantum codes are well suited for fault-tolerant quantum computation \cite{FT1, FT2, FT3} in two-dimensional qubit arrays constrained by short-range interaction. One of their topological features is that the number of encoded qubits depends only on the topology of the embedding surface but not on the code size. For such codes, the observables measured in order to identify errors are very close to local. They can be read out efficiently using local and nearest-neighbor quantum gates. 

The first family of topological quantum codes proposed were Kitaev's surface codes \cite{Kit1}. They combine many advantages, namely a high error threshold \cite{dennis02}, very low weight of stabilizer operators, an efficient decoding algorithm \cite{dennis02,Edmonds}, and the flexibility to implement encoded quantum gates in a topological fashion by code deformations \cite{RH, BombinX, BombinTwist, FowlerDef}. A reason why one might want to go beyond surface codes is to reduce the operational cost of fault-tolerant quantum computation in two-dimensional qubit arrays. For surface codes, only a small set of gates can be realized transversally or by code deformation. All other gates invoke the procedure of magic state distillation \cite{BK2} which, although it does not dominate the overhead scaling, increases the operational cost.

A further family of topological quantum codes are color codes \cite{bombin06, bombin07}. For them, a larger group of quantum gates can be implemented in a transversal or topological fashion, potentially reducing the operational overhead. For the color codes, the error-threshold turns out to be as high as for the surface code if  the syndrome measurement is assumed perfect \cite{KG1}. If error in the syndrome measurement is taken into account, surprisingly,  color codes perform better than the surface codes \cite{KG2}. However, for a full analysis in which the error-correction procedure is broken down into gates, each with its own error, the error-threshold for the color codes seems to have the smaller value \cite{Land}. 

Error-correction for the color codes can be cast as a matching problem on a hypergraph \cite{wang10}, for which---unlike the matching problem on graphs \cite{Edmonds} that applies to the surface codes \cite{dennis02}---no polynomial time algorithm is known. Still, approximate polynomial time solutions can be found \cite{wang10} that lead to a good error threshold. A different approach is taken in \cite{bombin11} where the color code on a square-octagon lattice is mapped onto two copies of a surface code, and the surface codes are then decoded in a recursive fashion \cite{poulin10} somewhat analogous to concatenated codes. Message passing is an integral part of the algorithm in \cite{poulin10}.

In this paper, we present an alternate decoding algorithm for color codes, and apply it to the hexagonal lattice for which no threshold value has yet been reported. We find an error threshold of 7.8\% for an error model with independent dephasing and spin flip errors. Like the algorithm in \cite{poulin10}, the present algorithm is recursive and employs message passing, but it decodes the color code directly without mapping to two copies of a surface code.

\section{Background}

\subsection{Topological color codes}
There are two main ingredients for constructing a color code: 
i) a trivalent graph
ii) an embedding of the graph on a surface so that its faces are  3-colorable. The stabilizer generators of the code are defined as 
\begin{eqnarray}
B_f^\sigma = \prod_{i\in f} \sigma_f, \label{eq:stabGenTCC}
\end{eqnarray}
where $f$ is a face of the embedding.

\begin{center}
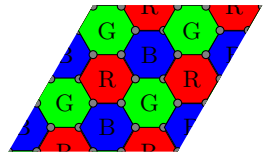
\begin{figure}[htb]
\begin{tikzpicture}[scale=.75]
\clip (-0.25,-2.15)--(1.25,0.45)--(4.25,0.45) --(2.75,-2.15);
\foreach \a in {-1,...,5}{
\node[regular polygon, regular polygon sides=6, minimum size=2, pattern=dots,draw,fill=red] at (\a*1.5+0.75,0.43) {R}; 
\node[regular polygon, regular polygon sides=6, minimum size=2, pattern=crosshatch dots, draw,fill=green] at (\a*1.5,0) {G}; 
\node[regular polygon, regular polygon sides=6, minimum size=2,draw,fill=blue] at (\a*1.5+0.75,-0.43) {B}; 
\node[regular polygon, regular polygon sides=6, minimum size=2, pattern =dots, draw,fill=red] at (\a*1.5,-0.86) {R};
\node[regular polygon, regular polygon sides=6, minimum size=2, pattern=crosshatch dots,  draw,fill=green] at (\a*1.5+0.75,-1.29) {G}; 
\node[regular polygon, regular polygon sides=6, minimum size=2, draw,fill=blue] at (\a*1.5,-1.72) {B}; 
\node[regular polygon, regular polygon sides=6, minimum size=2, pattern =dots, draw,fill=red] at (\a*1.5+0.75,-2.15) {R};
\node[regular polygon, regular polygon sides=6, minimum size=2, pattern=crosshatch dots, draw,fill= green] at (\a*1.5,-2.58) {G}; 

\draw (\a*1.5-0.25,0.43) [fill=gray] circle (2pt);
\draw (\a*1.5+0.25,0.43) [fill=gray] circle (2pt);

\draw (\a*1.5+0.5,0) [fill=gray] circle (2pt);
\draw (\a*1.5+1,0) [fill=gray] circle (2pt);

\draw (\a*1.5-0.25,-0.43) [fill=gray] circle (2pt);
\draw (\a*1.5+0.25,-0.43) [fill=gray] circle (2pt);

\draw (\a*1.5+0.5,-0.86) [fill=gray] circle (2pt);
\draw (\a*1.5+1,-0.86) [fill=gray] circle (2pt);

\draw (\a*1.5-0.25,-1.29) [fill=gray] circle (2pt);
\draw (\a*1.5+0.25,-1.29) [fill=gray] circle (2pt);

\draw (\a*1.5+0.5,-1.72) [fill=gray] circle (2pt);
\draw (\a*1.5+1,-1.72) [fill=gray] circle (2pt);

\draw (\a*1.5-0.25,-2.15) [fill=gray] circle (2pt);
\draw (\a*1.5+0.25,-2.15) [fill=gray] circle (2pt);

}
\end{tikzpicture}
\caption{Topological color code on a torus}\label{fig:tccOntorus}
\end{figure}
\end{center}

In this paper we consider the color code on a hexagonal lattice
embedded on a torus.  The same ideas can be applied to other color codes as well. We choose the embedding shown in Fig.~\ref{fig:tccOntorus}. 
The resulting code has the parameters $[[n,4]]$, where $n$ is the number of the vertices in the 
graph. The logical operators are associated with the nontrivial cycles on the torus, and carry a color label. The torus supports four encoded qubits. Our choice of hexagonal lattice and its embedding on the torus gives a code with the parameters $[[18\cdot 4^m, 4, 2^{m+2}]]$, where $m$ can be any nonnegative integer.

\subsection{The dual lattice}
Although the algorithm can be designed for topological color codes on an arbitrary graph, the codes most likely to be used (especially in a fault tolerant setting) are those embedded in a regular, translation-invariant structure, i.e., a lattice . Three suitable lattices are known, namely the hexagonal lattice, square-octagon lattice (also called the truncated square tiling) and
the truncated trihexagonal lattice \cite{anderson11}. We restrict our attention to codes on regular lattices, in particular to the hexagonal lattice. The color code on this lattice 
on a torus has to satisfy certain restrictions on its length in order to be embedded
on the torus. Assuming that both the logical operators color code have the same weight,
we have $n=2\cdot(3 x^m)^2$, where $m$ is a nonnegative integer and $x\geq 2$.
For the purposes of decoding it is helpful to 
represent the color code on the dual lattice. 

\begin{center}
\begin{figure}[htb]
\begin{tikzpicture}[scale=.75] 
\draw[ thick,color=black,step=1cm] (0,0) grid (6,6); 
\foreach \x in {1,...,6}{
	\foreach \y in {1,...,6}{
		\draw[ thick,color=black] (\x,\y-1)--(\x-1,\y);
	}
}

\foreach \x in {0,...,1}{
	\draw[pattern=dots] (\x*3,0)--(3*\x+1,0)--(3*\x+1,1)--(3*\x,1)--(3*\x,0);
	\draw[pattern=dots] (3*\x+1,1)--(3*\x+2,1)--(3*\x+3,0)--(3*\x+2,0)--(3*\x+1,1);
}

\foreach \x in {0,...,6}{
		\draw (\x,\x)  [fill=red] circle (3pt);
}
\foreach \x in {0,...,5}{
		\draw (\x+1,\x)  [fill=blue] circle (3pt);
		\draw (\x,\x+1)  [fill=green] circle (3pt);
}

\foreach \x in {0,...,4}{
		\draw (\x+2,\x)  [fill=green] circle (3pt);
		\draw (\x,\x+2)  [fill=blue] circle (3pt);
}
\foreach \x in {0,...,3}{
		\draw (\x+3,\x)  [fill=red] circle (3pt);
		\draw (\x,\x+3)  [fill=red] circle (3pt);
}
\foreach \x in {0,...,2}{
		\draw (\x+4,\x)  [fill=blue] circle (3pt);
		\draw (\x,\x+4)  [fill=green] circle (3pt);
}
\foreach \x in {0,...,1}{
		\draw (\x+5,\x)  [fill=green] circle (3pt);
		\draw (\x,\x+5)  [fill=blue] circle (3pt);
}
\draw (6,0) [fill=red] circle(3pt);
\draw (0,6) [fill=red] circle(3pt);
\end{tikzpicture}
\caption{Hexagonal color code on a dual lattice (opposite sides of the lattice are identified). In this picture each vertex corresponds to a check and 
each (triangular) face  corresponds to a qubit. Further, logical errors appear as nontrivial cycles of faces. 
One such operator is highlighted by the shaded faces.}\label{fig:tcc72}
\end{figure}
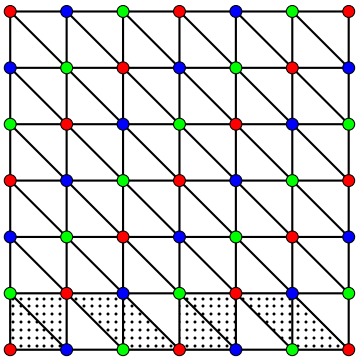
\end{center}

%\section{An Iterative Decoding Algorithm for Color Codes}

\section{The decoding algorithm---hard decision version}

\subsection{Rescaling} 
The main technique at work in the present algorithm is a rescaling transformation in which the code  is mapped into a smaller copy of itself.

The first step in the decoding the color code is to divide the code into
smaller code blocks each of which can be decoded using an efficient decoder.
For instance, the $[[72,4,8]]$ code shown in Fig.~\ref{fig:tcc72} can be divided as follows. 

\begin{center}
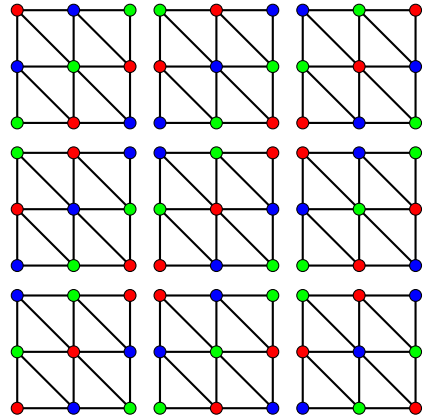
\begin{figure}[htb]
\begin{tikzpicture}[scale=.75] 

\newcommand\rCell{
\draw[ thick,color=black,step=1cm] (0,0) grid (2,2); 
\draw [thick, color=black] (1,0)--(0,1);
\draw [thick, color=black] (1,2)--(2,1);
\draw [thick, color=black] (0,2)--(2,0);
\foreach \x in {0,...,2}{
		\draw (\x,\x)  [fill=red] circle (3pt);
}
\foreach \x in {0,...,1}{
		\draw (\x+1,\x)  [fill=blue] circle (3pt);
		\draw (\x,\x+1)  [fill=green] circle (3pt);
}
\draw (2,0) [fill=green] circle( 3pt);
\draw (0,2) [fill=blue] circle( 3pt);
}

\newcommand\gCell{
\draw[ thick,color=black,step=1cm] (0,0) grid (2,2); 
\draw [thick, color=black] (1,0)--(0,1);
\draw [thick, color=black] (1,2)--(2,1);
\draw [thick, color=black] (0,2)--(2,0);
\foreach \x in {0,...,2}{
		\draw (\x,\x)  [fill=green] circle (3pt);
}
\foreach \x in {0,...,1}{
		\draw (\x+1,\x)  [fill=red] circle (3pt);
		\draw (\x,\x+1)  [fill=blue] circle (3pt);
}
\draw (2,0) [fill=blue] circle( 3pt);
\draw (0,2) [fill=red] circle( 3pt);
}

\newcommand\bCell{
\draw[ thick,color=black,step=1cm] (0,0) grid (2,2); 
\draw [thick, color=black] (1,0)--(0,1);
\draw [thick, color=black] (1,2)--(2,1);
\draw [thick, color=black] (0,2)--(2,0);
\foreach \x in {0,...,2}{
		\draw (\x,\x)  [fill=blue] circle (3pt);
}
\foreach \x in {0,...,1}{
		\draw (\x+1,\x)  [fill=green] circle (3pt);
		\draw (\x,\x+1)  [fill=red] circle (3pt);
}
\draw (2,0) [fill=red] circle( 3pt);
\draw (0,2) [fill=green] circle( 3pt);
}

\foreach \x in {-1,...,1}{	
\begin{scope}[xshift=(2*\x+0.4*\x)*30,yshift=(2*\x+0.4*\x)*30]
   \rCell
\end{scope}
}

\foreach \x in {0,...,1}{	
\begin{scope}[xshift=(2*\x+0.4*\x)*30,yshift=(2*\x-2.4+0.4*\x)*30]
   \gCell
\end{scope}
}

\foreach \x in {0,...,1}{	
\begin{scope}[xshift=(2*\x-2.4+0.4*\x)*30,yshift=(2*\x+0.4*\x)*30]
   \bCell
\end{scope}
}

\foreach \x in {0,...,0}{	
\begin{scope}[xshift=(2*\x+2.4+0.4*\x)*30,yshift=(2*\x-2.4+0.4*\x)*30]
   \bCell
\end{scope}
}

\foreach \x in {0,...,0}{	
\begin{scope}[xshift=(2*\x-2.4+0.4*\x)*30,yshift=(2*\x+2.4+0.4*\x)*30]
   \gCell
\end{scope}
}
\end{tikzpicture}
\caption{Dividing the lattice\label{latticeA}}
\end{figure}
\end{center}
Each block can be decoded separately and then replaced by a fewer number of logical qubits 
to give rise to another color code of smaller size, which can then be decoded using the same process.
The checks in each block can be distinguished into three classes: i) those that are interior to the block 
ii)  those that are on the boundary and shared between two blocks and iii) those on the corner and 
shared between more than one block. 

In decoding the block individually, we use the checks interior to the block and those on the boundary
but not those on the corners. These are used in the next level of decoding. The errors made in the 
lower levels of decoding are corrected in subsequent levels using these remaining checks. 

Recall that the qubit corresponds to a triangle. Thus, when we replace a set of qubits by 
another logical qubit, we expect it to emulate a qubit in the following sense: A single error on a 
qubit flips all the three checks on the corners of the qubit. Similarly, the rescaled qubit must be 
identified with a triangle whose corner checks can be flipped by a logical operator. 
The smallest such cell is as shown in Fig~\ref{fig:rescaling}.
\begin{center}
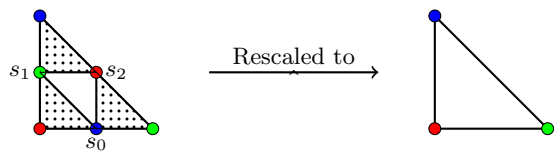
\begin{figure}[htb]
\begin{tikzpicture}[scale=.75,domain=0:2] 
%\draw[thin,color=gray!10!white,step=.5cm] (0,0) grid (4.4,3.4); 
\draw[thick,color=black] (0,0) -- (2,0) -- (0,2)-- (0,0) ; %plot coordinates {(2,0) (0,2)};
\draw[color=black] (1,0) node[below] {$s_0$};
\draw[color=black] (1,1) node[right] {$s_2$};
\draw[color=black] (0,1) node[left] {$s_1$};
\draw[thick,color=black] (1,0) -- (1,1) -- (0,1)-- (1,0) ; %plot coordinates {(2,0) (0,2)};
\draw[ thick, ->] (3,1)--(6.0,1) ;
\draw[fill=red] (1,1) circle (3pt );
\draw[fill=blue] (1,0) circle (3pt );
\draw[fill=green] (0,1) circle (3pt );

\draw[pattern=dots] (0,0)--(1,0)--(0,1);
\draw[pattern=dots] (1,0)--(2,0)--(1,1);
\draw[pattern=dots] (0,1)--(1,1)--(0,2);

\draw [->] (4.5,1) node[above]{Rescaled to};
\draw[thick,color=black,] (7,0) -- (9,0) -- (7,2)-- (7,0) ; %plot coordinates {(2,0) (0,2)};
		\draw (0,0)  [fill=red] circle (3pt);
		\draw (0,2)  [fill=blue] circle (3pt);
		\draw (2,0)  [fill=green] circle (3pt);
		\draw (7,0)  [fill=red] circle (3pt);
		\draw (7,2)  [fill=blue] circle (3pt);
		\draw (9,0)  [fill=green] circle (3pt);

\end{tikzpicture}
\caption{Rescaling a group of qubits.}\label{fig:rescaling}
\end{figure}
\end{center}

The shaded qubits in Fig.~\ref{fig:rescaling} correspond to a logical operator that emulates 
a bit flip on the rescaled qubit. This operator flips the checks on the corners of the 
rescaled qubit. The collection of these qubits then behaves as a single logical qubit with this 
logical  operator. 

\begin{center}
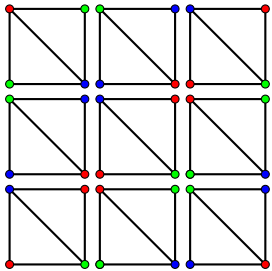
\begin{figure}[htb]

\begin{tikzpicture}[scale=.5] 
\foreach \x in {1,...,3}{
	\foreach \y in {1,...,3}{
		\draw[ thick,color=black,step=1cm] (2*\x+0.4*\x,2*\y+0.4*\y) rectangle (2*\x+2+0.4*\x,2*\y+2+0.4*\y); 
		\draw[ thick,color=black] (2*\x+0.4*\x,2*\y+0.4*\y+2)--(2*\x+0.4*\x+2,2*\y+0.4*\y);
		}
}

\foreach \x in {1,...,3}{
		\draw (2*\x+0.4*\x,2*\x+0.4*\x)  [fill=red] circle (3pt);
		\draw (2*\x+2+0.4*\x,2*\x+2+0.4*\x)  [fill=red] circle (3pt);
		\draw (2*\x+2+0.4*\x,2*\x+0.4*\x)  [fill=green] circle (3pt);
		\draw (2*\x+0.4*\x,2*\x+2+0.4*\x)  [fill=blue] circle (3pt);

}

\foreach \x in {1,...,2}{
		\draw (2*\x+0.4*\x,2*\x+2+0.4*\x+0.4)  [fill=blue] circle (3pt);
		\draw (2*\x+2+0.4*\x,2*\x+4+0.4*\x+0.4)  [fill=blue] circle (3pt);
		\draw (2*\x+2+0.4*\x,2*\x+2+0.4*\x+0.4)  [fill=red] circle (3pt);
		\draw (2*\x+0.4*\x,2*\x+4+0.4*\x+0.4)  [fill=green] circle (3pt);
}

\foreach \x in {2,...,3}{
		\draw (2*\x+0.4*\x,2*\x-2+0.4*\x-0.4)  [fill=green] circle (3pt);
		\draw (2*\x+0.4*\x,2*\x-2+0.4*\x-0.4)  [fill=green] circle (3pt);

		\draw (2*\x+0.4*\x,2*\x+0.4*\x-0.4)  [fill=red] circle (3pt);
		\draw (2*\x+2+0.4*\x,2*\x+0.4*\x-0.4)  [fill=green] circle (3pt);
		\draw (2*\x+2+0.4*\x,2*\x-2+0.4*\x-0.4)  [fill=blue] circle (3pt);
}

\foreach \x in {1,...,1}{
		\draw (2*\x+2+0.4*\x,2*\x+4+0.4*\x+0.8)  [fill=blue] circle (3pt);
		\draw (2*\x+0.4*\x,2*\x+4+0.4*\x+0.8)  [fill=green] circle (3pt);
		\draw (2*\x+2+0.4*\x,2*\x+6+0.4*\x+0.8)  [fill=green] circle (3pt);
		\draw (2*\x+0.4*\x,2*\x+6+0.4*\x+0.8)  [fill=red] circle (3pt);
}

\foreach \x in {3,...,3}{
		\draw (2*\x+0.4*\x,2*\x-4+0.4*\x-0.8)  [fill=blue] circle (3pt);
		\draw (2*\x+2+0.4*\x,2*\x-4+0.4*\x-0.8)  [fill=red] circle (3pt);
		\draw (2*\x+0.4*\x,2*\x-2+0.4*\x-0.8)  [fill=green] circle (3pt);
		\draw (2*\x+2+0.4*\x,2*\x-2+0.4*\x-0.8)  [fill=blue] circle (3pt);
}
\end{tikzpicture}
\caption{Rescaled lattice, after applying the rescaling transformation of Fig.~\ref{fig:rescaling} to the lattice of Fig.~\ref{latticeA}.} 
\end{figure}
\end{center}

There are a restrictions on the scale factor for the rescaling. For instance, the cell 
$$
%\begin{center}
%\begin{figure}[htb]
\begin{tikzpicture}[scale=.5,domain=0:2] 
\draw[thick,color=black,] (0,0) -- (3,0) -- (0,3)-- (0,0) ; %plot coordinates {(2,0) (0,2)};
\draw[thick,color=black,] (1,0) -- (1,1) -- (0,1)-- (1,0) ; %plot coordinates {(2,0) (0,2)};
\draw[thick,color=black,] (1,1) -- (1,2) -- (0,2)-- (1,1) ; %plot coordinates {(2,0) (0,2)};
\draw[thick,color=black,] (2,0) -- (2,1) -- (1,1)-- (2,0) ; %plot coordinates {(2,0) (0,2)};
		\draw (0,0)  [fill=red] circle (3pt);
		\draw (1,0)  [fill=blue] circle (3pt);
		\draw (2,0)  [fill=green] circle (3pt);
		\draw (3,0)  [fill=red] circle (3pt);

		\draw (0,1)  [fill=green] circle (3pt);
		\draw (1,1)  [fill=red] circle (3pt);
		\draw (2,1)  [fill=blue] circle (3pt);

		\draw (0,2)  [fill=blue] circle (3pt);
		\draw (1,2)  [fill=green] circle (3pt);

		\draw (0,3)  [fill=red] circle (3pt);

\end{tikzpicture}
%\caption{A cell which cannot be rescaled}
%\end{figure}
%\end{center}
$$
is not legitimate for rescaling. 

The restriction arises for the following reason. In each cell, there are more error locations than syndrome bits. Therefore, each error can locally be corrected only up to an undetectable error $X$ that remains for the next  level. For the rescaling to work, the error  must  act as an encoded error on the rescaled face. That is, $X$ either flips all corner syndromes or none. This observation leads to    
\begin{lemma}\label{lm:restrictions}
Suppose that we have a hexagonal color code represented on its dual lattice. Then an
equilateral triangular cell in the (dual) lattice of side $n$
containing $n^2$ qubits and  $n(n+1)/2$ checks can be rescaled if and only if 
$n\not\equiv 0 \bmod 3$.
\end{lemma}
\begin{proof}
Suppose that the cell has a  length that is not a multiple of 3. The 
rescaled cell must have the following properties. There must exist an operator which
acts like the bit flip operator on a single qubit, it must therefore be able to flip 
all the corner checks. The crucial observation that we need is that the two checks of the
same color can be connected by error operators as shown in Fig.~\ref{fig:2-chains}.

\begin{center}
\begin{figure}[htb]
\begin{tikzpicture}[scale=.75] 
\clip (0.5,0.5) rectangle (4.5,4.5);
\draw[ color=black,step=1cm] (0,0) grid (6,6); 

\foreach \x in {1,...,6}{
	\foreach \y in {1,...,6}{
		\draw[ color=black] (\x,\y-1)--(\x-1,\y);
	}
}

%\draw [ultra thick, color=red] (2,2) .. controls (2,3) and (4,3) --(3,3)--(4,1);
\draw [ultra thick, color=red]  (2,2) .. controls (2.5,2) and (2.5,3) .. (3,3) .. controls (3.5,2.5) and (3.5,1.5) .. (4,1);

\draw [ultra thick, color=blue]  (1,3) .. controls (1.5,3) and (1.5,4) .. (2,4) .. controls (2.5,4) and (3.5,3.5) .. (4,3);
	\draw[pattern=dots, pattern color=red] (2,2) rectangle(3,3);
	\draw[pattern=dots, pattern color=red] (3,3)--(3,2)--(4,1)--(4,2)--(3,3);

%	\draw[pattern=dots, pattern color=red] (2,2) rectangle(3,3);
%	\draw[pattern=dots, pattern color=red] (3,3)--(3,2)--(4,1)--(4,2)--(3,3);

	\draw[pattern=dots, pattern color=blue] (1,3) rectangle(2,4);
	\draw[pattern=dots, pattern color=blue] (2,4)--(3,3)--(4,3)--(3,4)--(2,4);

%	\draw[pattern=crosshatch, pattern color=blue] (1,3) rectangle(2,4);
%	\draw[pattern=crosshatch, pattern color=blue] (2,4)--(3,3)--(4,3)--(3,4)--(2,4);

\foreach \x in {0,...,6}{
		\draw (\x,\x)  [fill=red] circle (3pt);
}
\foreach \x in {0,...,5}{
		\draw (\x+1,\x)  [fill=blue] circle (3pt);
		\draw (\x,\x+1)  [fill=green] circle (3pt);
}

\foreach \x in {0,...,4}{
		\draw (\x+2,\x)  [fill=green] circle (3pt);
		\draw (\x,\x+2)  [fill=blue] circle (3pt);
}
\foreach \x in {0,...,3}{
		\draw (\x+3,\x)  [fill=red] circle (3pt);
		\draw (\x,\x+3)  [fill=red] circle (3pt);
}
\foreach \x in {0,...,2}{
		\draw (\x+4,\x)  [fill=blue] circle (3pt);
		\draw (\x,\x+4)  [fill=green] circle (3pt);
}
\foreach \x in {0,...,1}{
		\draw (\x+5,\x)  [fill=green] circle (3pt);
		\draw (\x,\x+5)  [fill=blue] circle (3pt);
}
\draw (6,0) [fill=red] circle(3pt);
\draw (0,6) [fill=red] circle(3pt);
\end{tikzpicture}
\caption{2-chains on the color code}\label{fig:2-chains}
\end{figure}
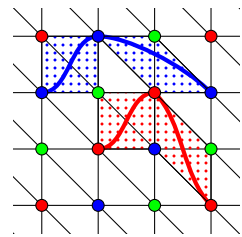
\end{center}

Because of the 3-colorability of the color code, the checks on the
three corners of the cell are of different color. Now assume that there is a single qubit 
error in the interior of the cell, (the checks could be on the boundary but the qubit
itself is in the cell). Then this error causes three checks (each of different color) in its corners to flip.

\begin{center}
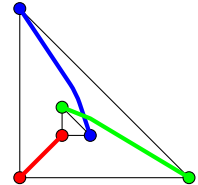
\begin{figure}[htb]
\begin{tikzpicture}[scale=.75] 
\clip (0.5,0.5) rectangle (4.5,4.5);
%\draw[ color=black,step=1cm] (0,0) grid (6,6); 
\draw (1,1)--(4,1)--(1,4)--(1,1);

\draw (1.75,1.75)--(2.25,1.75)--(1.75,2.25)--(1.75,1.75);

\draw (1,1) [fill=red] circle(3pt);
\draw (1,4) [fill=blue] circle(3pt);
\draw (4,1) [fill=green] circle(3pt);

\draw (1.75,1.75) [fill=red] circle(3pt);
\draw (2.25,1.75) [fill=blue] circle(3pt);
\draw (1.75,2.25) [fill=green] circle(3pt);

\draw [ultra thick,  red] (1.75,1.75)--(1,1);
\draw [ultra thick,rounded corners,blue](2.25,1.75)--(2,2.5)--(1,4);
\draw [ultra thick,green] (1.75,2.25)--(2.25,2.05)--(4,1);

\end{tikzpicture}
\caption{Logical operator when $n\not\equiv 0 \bmod 3$}\label{fig:logicalOperator}
\end{figure}
\end{center}

Since all the cells are connected, it is possible to move these checks to the corners of the cell.
The operator that achieves this acts as the logical  operator for the cell.

If on the other hand, the side is a multiple of 3, all the three corners of the triangle 
are of the same color. One 2-chain can be used to flip two corners. However, the third corner
will either lead to a flip of one or more check in the interior of the cell if its support
is restricted to the cell. 
So it is not possible for any error pattern to simultaneously flip 
all the three corners. 
\end{proof}
The above result can be generalized to color codes that are not based on regular lattices. 

\subsection{Splitting the syndrome}\label{ssec:splitHD}

The motivation behind the division of the lattice into smaller blocks is to decode each cell
locally. Since the checks are shared between adjacent cells, this implicitly requires us 
to partition the measured syndrome between the adjacent cells so that we can proceed
with the local decoding of cells. Each (non-corner) syndrome in the boundary of the cells must
therefore be split between the adjacent cells. 
For concreteness let us consider the 
following situation shown in Fig.~\ref{fig:split}. Supposing a shared syndrome were measured to be $s$. 
Then the block on the top and bottom can 
contribute $s^{t}$ and $s^{l}$ such that $s^{t}+s^{l}=s$. 
\begin{center}
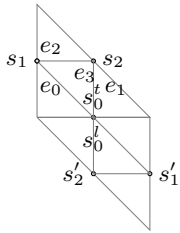
\begin{figure}[htb]
\begin{tikzpicture}[scale=0.75,domain=0:2] 
%\draw[thin,color=gray!10!white,step=.5cm] (0,0) grid (4.4,3.4); 
\draw[thin,color=gray] (0,0) -- (2,0) -- (0,2)-- (0,0) ; %plot coordinates {(2,0) (0,2)};
\draw[color=black] (1,0) node[above] {$s_0^t$};
\draw[color=black] (1,1) node[right] {$s_2$};
\draw[color=black] (0,1) node[left] {$s_1$};
\draw[fill=gray!50!white] (1,1) circle (1pt );
\draw[fill=gray!50!white] (1,0) circle (1pt );
\draw[fill=gray!50!white] (0,1) circle (1pt );
\draw[thin ,color=gray] (1,0) -- (1,1) -- (0,1)-- (1,0) ; %plot coordinates {(2,0) (0,2)};

\draw[color=black] (0.25,0.25) node[above] {$e_0$};
\draw[color=black] (1.4,.25) node[left,above] {$e_1$};
\draw[color=black] (0.25,1.5) node[below] {$e_2$};
\draw[color=black] (.5,.75) node[right] {$e_3$};
\draw[thin,color=gray] (0,0) -- (2,0) -- (2,-2)-- (0,0) ; %plot coordinates {(2,0) (0,2)};
\draw[color=black] (1,0) node[below] {$s_0^l$};
\draw[color=black] (1,-1) node[left] {$s_2'$};
\draw[color=black] (2,-1) node[right] {$s_1'$};
\draw[fill=gray!30!white] (1,-1) circle (1pt );
\draw[fill=gray!30!white] (2,-1) circle (1pt );
\draw[fill=gray!30!white] (0,1) circle (1pt );
\draw[thin,color=gray] (1,-1) -- (2,-1) -- (1,0)-- (1,-1) ; %plot coordinates {(2,0) (0,2)};
%\draw[color=black] (1.75,-0.25) node[below] {$f_0$};
%\draw[color=black] (1.25,-.25) node[right, below] {$f_1$};
%\draw[color=black] (1.75,-1.5) node[above] {$f_2$};
%\draw[color=black] (.25,-.25) node[right] {$f_3$};
\end{tikzpicture}
\caption{Splitting the shared syndromes}\label{fig:split}
\end{figure}
\end{center}

For simplicity consider the following heuristic. Assume that we randomly choose the initial split of 
the check. Let $p(e|s_0^t s_1 s_2)$ denote the probability of errors that are consistent with 
the syndrome. Then the probability that $s^t=1$ is given by 
\ben
p(s^t|s_1s_2) =\frac{p(e|s^t, s_1 s_2)}{p(e|s^t=1, s_1s_2)+p(e|s^t=0, s_1s_2)}.\label{eq:hdupdate}
\een
Similarly, we can compute the probability of the lower cell syndrome $s^l$ as 
\ben
p(s^l|s_1's_2') =\frac{p(e|s^l, s_1' s_2')}{p(e|s^l=1, s_1's_2')+p(e|s^l=0, s_1's_2')}.\label{eq:hdupdate2}
\een
Next the cell on the top and bottom exchange messages so that their estimates for the split probabilities are consistent. 
We then enforce the condition that $s=s^t+s^l$, which leads to an update of the 
$p(s^t)$ and $p(s^l)$ as follows. %We use the same notation for the updated probabilities
Under the assumption that $s=0$, we compute
$p(s^t|s_1s_2s_1's_2')$ as
\ben
 \frac{p(s^t|s_1s_2)p(s^l|s_1's_2')}{p(s^t)p(s^l|s_1's_2')+ (1-p(s^t|s_1s_2))(1-p(s^l|s_1's_2'))}.
 \label{eq:consistent0}
\een
 If $s=1$, then we compute
$p(s^t|s_1s_2s_1's_2')$ as
\ben
 \frac{p(s^t|s_1s_2)(1-p(s^l|s_1's_2'))}{p(s^t)(1-p(s^l|s_1's_2'))+ (1-p(s^t|s_1s_2))p(s^l|s_1's_2')}.
 \label{eq:consistent1}
\een
After the conditional split probabilities have been obtained through Eqs.~(\ref{eq:hdupdate}) - (\ref{eq:consistent1}) we make a hard decision for the syndrome split variable on each cell boundary in the lattice. We either choose the most likely splitting, or make a biased random choice with the respective split probability as prior. Then, we repeat the entire procedure until the pattern of syndrome splits reaches a stable configuration. In practice, we iterate for a predetermined number of rounds.

\subsection{Decoding the basic cell}
The basic cell is decoded using a complete lookup table which is an optimal decoder. Given the split probabilities after the previous step, we make a splitting for all the shared syndromes. Every cell in the code can now be decoded locally with the syndrome
internal to the cell and the syndrome split on the boundary. The effect of this local error
correction is that all syndromes in the interior and the boundaries of the cells are accounted
for. Only the syndromes in the corners of the cells are not accounted for. These checks 
will be used at the next level, to account for any errors in the estimates made by the
decoding in the present stage. Before we can move to the next stage however, we need an error model.
We next show how to derive this error model for the next stage, but first a few remarks on 
the complexity of decoding the basic cells. 

Because the cell is of a constant size
as the length of the code increases there is no increase in the complexity of this decoder 
with the length. The overhead comes in the form of the number of cells which goes polynomially
in the length of the code. For code of length $n=18\cdot x^{2m}$, with a cell size of $x$,
there are $O(n/x^2)$ cells. Therefore the decoding complexity is $O(n)$ at each level. 
The overall complexity including the rescaling at each stage we obtain the complexity as 
$O(n\log_2n)$.

\subsection{Rescaling the lattice }
For the iterative procedure to work, we need the error model at each level of rescaling. 
This is obtained as follows. Based on the split probabilities we obtain a splitting of 
all the shared syndromes, which we shall henceforth call the {\em reference splitting}. This
can be made randomly or simply choosing the most likely split for each shared syndrome. 
The latter is deterministic---the same set of split probabilities will return the same
reference splitting. 

Assume that we have a certain splitting for the block being 
decoded. Then the block decoder returns an estimate for this syndrome and also a probability
of error in the estimate.
Now suppose that the syndrome is $s_0s_1s_2$, then the rescaling algorithm makes the correction that is given by canonical estimate corresponding to $s_0s_1s_2$ and returns
the probability of error in the rescaled qubit as 
\ben
p_{\text{rescaled},s_0s_1s_2}^{\text{hard}} &= &\frac{p(e+\overline{X}|s_0s_1s_2)}{p(e|s_0s_1s_2)+p(e+\overline{X}|s_0s_1s_2)},\label{eq:rescaledP}
\een
where $e$ is the canonical estimate  and $\bar{e}=e+\overline{X}$, the complementary error. We pass on this error probability as prior to the next level.
 
 \section{Refinements}

 \subsection{Soft splitting of syndrome}

 The hard-decision heuristic for finding the configuration of syndrome splits described in  Section~\ref{ssec:splitHD} discards much available information. Namely, in each round of the iteration, split probabilities are being computed but immediately rounded to 0 or 1. A better approach is the following soft decision version, in which a recursion relation is set up for the split probabilities. The split probabilities are computed as 
 \ben
 p(s^t) = \sum_{s_1s_2}\frac{p(e|s^t, s_1 s_2) p(s_1s_2)}{p(e|s^t=1, s_1 s_2)+p(e|s^t=0, s_1s_2)}.
 \label{eq:splitRule0}\\
 p(s^l) = \sum_{s_1's_2'}\frac{p(e|s^l, s_1' s_2') p(s_1's_2')}{p(e|s^l=1, s_1's_2')+p(e|s^l=0, s_1's_2')}.\label{eq:splitRule1}
 \een
 Once again, consistency with respect to the measured syndrome is enforced as in equations~\eqref{eq:consistent0}~and~\eqref{eq:consistent1}. For instance, if the measured syndrome is $s=0$ we obtain the following
 update for the top split probability, 
 \ben
p(s^t) & \mapsto & \frac{p(s^t)(1-p(s^l))}{p(s^t)(1-p(s^l))+ (1-p(s^t))p(s^l)}.\label{eq:splitRuleConsistent}
\een
For the initial split of  $p(s^t)$ and $p(s^l)$,  we compute 
$p_l $ as probability of odd number of errors in $e_0e_1e_3$ and $p_r$ as the probability of an
odd number of errors in the lower cell. Then we compute $p(s^t)$ and $p(s^l)$ from $p_l$ and 
$p_r$ conditioned on the measured syndrome as in equations~\eqref{eq:consistent0}~and~\eqref{eq:consistent1}. 

\subsection{Soft rescaling of the lattice}

As in the syndrome splitting case, the hard decision approach is not optimal. 
Since we do not know for sure which splitting is the correct one, it is preferable to use all the information available to perform the rescaling.
This raises the problem of consistently combining the error probabilities from each splitting. 

The key idea in making this combination is the fact that a stabilizer generator has the effect of
flipping the splitting on that check. This provides  us with a means to relate canonically the errors
across different splittings. Suppose that we have an estimate $e$ for a splitting $s$. If we changed
the splitting to some other splitting $\tilde{s}$, then the canonical estimate for this splitting
is the same as the one obtained by applying half stabilizers that flip $s$ to $\tilde{s}$.

Let us elaborate on this in more detail.
The choice for the estimate  is follows. Let $S_i$ be the check operator on the $i$th
check node (of the $2\times 2$ triangular cell shown in Fig.~\ref{fig:softRescaling}). Then the support of $S_i$ in the cell is given as
$S_i|_{\text{cell}}$. For instance, the 0th check has support on qubits   $\{0,1,3\}$, they are shaded in the
figure below.

\begin{center}
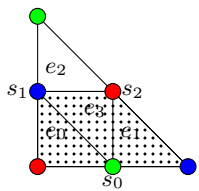
\begin{figure}[htb]
\begin{tikzpicture}[scale=1,domain=0:2] 
%\draw[thin,color=gray!10!white,step=.5cm] (0,0) grid (4.4,3.4); 
\draw[color=black] (0,0) -- (2,0) -- (0,2)-- (0,0) ; %plot coordinates {(2,0) (0,2)};
\draw[color=black] (1,0) node[below] {$s_0$};
\draw[color=black] (1,1) node[right] {$s_2$};
\draw[color=black] (0,1) node[left] {$s_1$};
\draw[color=black] (1,0) -- (1,1) -- (0,1)-- (1,0) ; %plot coordinates {(2,0) (0,2)};
\draw[pattern = dots] (0,0)--(2,0)--(1,1)--(0,1)--(0,0);

\draw[fill=red] (0,0) circle (3pt );
\draw[fill=blue] (2,0) circle (3pt );
\draw[fill=green] (0,2) circle (3pt );
\draw[fill=red] (1,1) circle (3pt );
\draw[fill=blue] (0,1) circle (3pt );
\draw[fill=green] (1,0) circle (3pt );
\draw[color=black] (0.25,0.25) node[above] {$e_0$};
\draw[color=black] (1.25,.25) node[left,above] {$e_1$};
\draw[color=black] (0.25,1.5) node[below] {$e_2$};
\draw[color=black] (.5,.75) node[right] {$e_3$};

\end{tikzpicture}
\caption{Rescaling }\label{fig:softRescaling}
\end{figure}
\end{center}

The crucial observation is that by applying $S_0$, we can flip the syndrome $s_0$. So
supposing we have a splitting $s_0s_1s_2$ with estimate $e_0 e_1 e_2 e_3$, to obtain the 
estimate for say 
$\bar{s}_0s_1s_2$, where $\bar{s}_0= s_0\oplus 1$ we can take the estimate of 
$s_0s_1s_2$ and then flip the qubits on the support of $S_0$ in the cell, thus we get 
$\bar{e}_0 \bar{e}_1 e_2 \bar{e}_3$ as the estimate for $\bar{s}_0s_1s_2$.
So letting $0000$ be the canonical estimate for the all zero syndrome we obtain the estimate
for $s_0s_1s_2=100$ as $e_0 e_1 e_2 e_3=1101$ and similarly for the rest. 

The rationale for this choice can be understood as follows. We can view any splitting that
differs from the reference splitting as having the same effect as the reference splitting 
if it can be obtained from the estimate from the reference 
splitting by applying a stabilizer operator. 

Now we can give a soft decision based rule for updating the rescaled error probabilities. 
Assume that a given cell has the reference splitting 
 $\tilde{s}_0 \tilde{s}_1\tilde{s}_2$. If this splitting is correct, then
we have the rescaled qubit's error probability as given in \eqref{eq:rescaledP}. But supposing that the correct splitting was $s_0s_1s_2$, then the
error in the rescaled qubit will have to be conditioned on i) $s_0s_1s_2$,
and ii) the estimate for the cell and the correction were applied under the assumption that 
$\tilde{s}_0 \tilde{s}_1\tilde{s}_2$ was the correct splitting for the cell.
%The weight of $s_0s_1s_2\oplus \tilde{s}_0 \tilde{s}_1\tilde{s}_2$ gives the number of errors in the splitting applied with respect to the correct splitting $s_0s_1s_2$.

Consider the stabilizer operator $S_d$ 
\ben 
S_d = \prod_{\stackrel{i=0} {i: s_i\oplus \tilde{s}_i=1}}^2 S_i
\een
Denote by $e_d=( S_d|_{\text{cell}}) $ the restriction of $S_d$ to the basic cell under consideration. 
The support of $e_d$   will determine
if the erroneous splitting that was applied requires a correction on the rescaled qubit. 
Let $\tilde{e}$ be the canonical error estimate associated with the splitting $s_0s_1s_2\oplus \tilde{s}_0 \tilde{s}_1\tilde{s}_2$. Then the error 
associated to $s_0s_1s_2$ which does not require a correction at the next level is given by
$\tilde{e} + e_d$.
Then the probability of error in the rescaled qubit is
given by 
\ben
p_{\text{rescaled},s_0s_1s_2}^{\text{hard}}  = \frac{p(e_d+\tilde{e}+\overline{X})}{p(e_d+\tilde{e}+\overline{X})+p(e_d+\tilde{e})}
\een
\ben
p_{\text{rescaled}}^{\text{soft}} &=&\sum_{s_0s_1s_2}p_{\text{rescaled},s_0s_1s_2}^{\text{hard}} p(s_0s_1s_2)\\
& =& \sum_{s_0s_1s_2}\frac{p(e_d+\tilde{e}+\overline{X}) p(s_0s_1s_2)}{p(e_d+\tilde{e}+\overline{X})+p(e_d+\tilde{e})}\label{eq:rescale}
\een

Let us explicitly write down this rule for the case of $2\times 2 $ cell.
\begin{equation}
\begin{array}{rcl}
p &=& \frac{p(1110)}{p(1110)+p(0000)}p(000)+ 
\frac{p(1001)}{p(1001)+p(0111)}p(001)+\\
&+&\frac{p(0101)}{p(0101)+p(1011)}p(010)+
\frac{p(0010)}{p(0010)+p(1100)}p(011)\\
&+& \frac{p(0011)}{p(0011)+p(1101)}p(100)+
\frac{p(0100)}{p(0100)+p(1010)}p(101)+\\
&+&\frac{p(1000)}{p(1110)+p(0110)}p(110)+
\frac{p(1111)}{p(1111)+p(0001)}p(111)
\end{array}
\end{equation}
Note that the rescaled error probability is independent of the reference splitting.

\subsection{Looking ahead at the corners}
For small cells, the algorithm has to slightly modified to account for certain malicious error patterns that would otherwise persist with appreciable probability at higher levels of the recursion. An example is
$$
%\begin{center}
%\begin{figure}[htb]
\begin{tikzpicture}[scale=.75] 
\draw[ thick,color=black,step=1cm] (0,0) grid (2,2); 
\foreach \x in {1,...,2}{
	\foreach \y in {1,...,2}{
		\draw[ thick,color=black] (\x,\y-1)--(\x-1,\y);
	}
}

\foreach \x in {0,...,0}{
	\draw[pattern=dots] (\x*3,0)--(3*\x+1,0)--(3*\x+1,1)--(3*\x,1)--(3*\x,0);
}

\foreach \x in {0,...,2}{
		\draw (\x,\x)  [fill=red] circle (3pt);
}
\foreach \x in {0,...,1}{
		\draw (\x+1,\x)  [fill=blue] circle (3pt);
		\draw (\x,\x+1)  [fill=green] circle (3pt);
}

		\draw (0,2)  [fill=blue] circle (3pt);
		\draw (2,0)  [fill=green] circle (3pt);

\end{tikzpicture}
%\caption{A malicious  error pattern}\label{fig:errPatterns}
%\end{figure}
%\end{center}
$$
Such patterns can be accounted for by modifying the qubit probabilities in the prior, by taking the corner syndromes into account ``ahead of time''. This will direct the decoding algorithm to correct these errors. %The heuristic that we used was as follows. 

Consider a corner syndrome. Each qubit in the
check belongs to a different cell. 

%\begin{figure}[htb]
$$
\begin{tikzpicture}%[scale=0.75]
\draw[ thin, gray ] (0,0)--(-1,0)--(-1,1)--(0,0);
\draw[ thin, gray ] (0,0)--(1,0)--(0,1)--(0,0);
\draw[ thin, gray ] (0,0)--(-1,1)--(0,1)--(0,0);
\draw[ thin, gray ] (0,0)--(-1,0)--(0,-1)--(0,0);
\draw[ thin, gray ] (0,0)--(0,-1)--(1,-1)--(0,0);
\draw[ thin, gray ] (0,0)--(1,-1)--(1,0)--(0,0);
\draw[ thin ]  (0,0)--(0,-2)--(2,-2)--(0,0);
\draw (0,0) [fill=gray] circle (3pt);

\draw (-0.5,0) node[left, below] {$p_1$};
\draw (-0.5,0) node[left, above] {$p_2$};
\draw (0,0.5) node[left] {$p_3$};
\draw (0,0.5) node[right] {$p_4$};
\draw (0.5,0) node[below] {$p_5$};
\draw (0,-0.5) node[right] {$p_6$};

\end{tikzpicture}
$$
%\end{figure}

We update the probabilities of these qubits as follows. 
Denote by $p_i=p_6$, the probability of error in the qubit interior to the cell.
Let $s$ be the corner syndrome. Then the updated
error probability of the qubit is given as 
\begin{eqnarray}
p_6 \mapsto \frac{p_ep_i}{p_ep_i+(1-p_e)(1-p_i)},\label{eq:cornerUpdate2}
\end{eqnarray}
where $p_e$ is the contribution of the qubits external to the cell
\begin{eqnarray}
p_e = \frac{1}{2} -(-1)^s\frac{1}{2} \prod_{j=1}^5 (1-2p_j).\label{eq:cornerUpdate1}
\end{eqnarray} 

\subsection{Belief propagation for prior update}
The performance of the algorithm can be improved by providing a more accurate estimate of the 
qubit error probabilities to the decoder. This can be achieved, as in the case of \cite{poulin10}
by using a standard belief propagation decoder to update the prior estimates for the qubit 
error probabilities. Unlike the toric codes, the color codes have a smaller girth and a larger number
of 4-cycles in their Tanner graph, therefore  it is not recommended to use too many iterations in the belief propagation algorithm to update the priors. 
%This leads to an improvement in the performance of the decoder. 

The complete algorithm for the decoder is now given in Algorithm~\ref{alg:tccDecoder}.

\begin{algorithm}[H]
\caption{{\ensuremath{\mbox{\scshape Decoding topological color code }}}}\label{alg:tccDecoder}
\begin{algorithmic}[1]
\REQUIRE {Syndrome, channel error rate.}
\STATE Update the initial error probability on the qubits.
\STATE For each syndrome in the boundary compute initial split probabilities based on the qubits in the cell as per the following rule
\begin{eqnarray*}
p=\frac{1}{2}-\frac{1}{2}\prod_i (1-2p_i).
\end{eqnarray*}
\STATE Send $p$ or $1-p$ to the neighboring cell depending on whether the shared syndrome is 0 or 1.
\STATE Enforce consistency on the split probabilities as per the following rule.
\begin{eqnarray*}
\frac{p_lp_r}{p_lp_r+(1-p_l)(1-p_r)}
\end{eqnarray*}
where $p_l $ is the split probability as estimate by the cell and the $p_r$ is the message
from the adjacent cell. 

\STATE Update the corner qubit probabilities based on the corner syndrome according
to \eqref{eq:cornerUpdate1}. Let $p_i$ be the probability that there is an error in the corner
qubit and $p_e$ be the probability that there is an error in the qubits external to the cell.
Then update the corner qubit probability by enforcing consistency on $p_i$ and $p_e$ as per
\eqref{eq:cornerUpdate2}.

\STATE Send messages to the neighbouring cells and estimate new split probabilities according
to equations~\eqref{eq:splitRule0},~\eqref{eq:splitRule1} and \eqref{eq:splitRuleConsistent}. 
\STATE Split the syndrome according to the split probabilities.
\STATE Locally estimate the error and correct. 
\STATE Compute the error probabilities for the next level of recursion as in equation~\eqref{eq:rescale}.
\STATE Repeat until the level 0.
\STATE Use a lookup table to estimate the final error.
\STATE Propagate the error to lower levels to obtain the final estimate. 
\end{algorithmic}
\end{algorithm}

\section{Conclusions}
\begin{figure}%[htb]
\includegraphics[width=7.5cm]{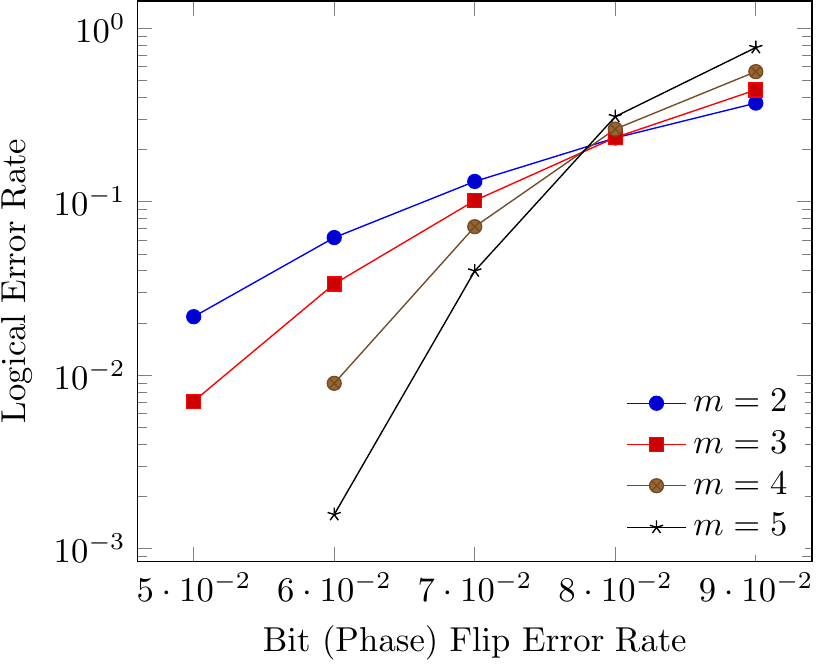}

\caption{Simulation results. The logical error rate after error-correction vs the physical error probability of independent  spin (phase) flip errors. Shown is the performance of $[[18\cdot 4^m, 4, 2^{m+2}]]$ codes 
for $m=2$, $3$, $4$ and $5$.}\label{fig:thrRes}
\end{figure}

We have presented a recursive decoding algorithm for the color code on the hexagonal lattice. It yields a  threshold value of 7.8\% for an error model of independent spin and phase flip errors, under the assumption that syndrome extraction is perfect. The results of our numerical simulation are shown in Fig.~\ref{fig:thrRes}. The recursion is based on a square cell of 8 qubits.

\noindent
\textbf{Acknowledgments}
This research is supported by grants from NSERC, MITACS and CIFAR.
{P.S. would like to thank David Poulin and Guillaume Duclos-Cianci for discussions.}

\end{document}